\documentclass[runningheads]{llncs}

\usepackage[utf8]{inputenc}
\usepackage[T1]{fontenc}
\usepackage{url}
\usepackage{amsfonts, amsmath, amssymb}
\usepackage{stix}
\usepackage{a4wide}
\usepackage{hyperref}
\usepackage{enumitem}
\usepackage{listings}
\usepackage[scaled=0.83]{beramono}
\lstdefinelanguage{pseudopython}{
  morekeywords={def,return,if,else}
}

\setlist[itemize]{nosep, noitemsep, topsep=0pt}
\setlist[enumerate]{nosep, noitemsep, topsep=0pt}

\hypersetup{
  colorlinks=true,
  urlcolor=blue,
  citecolor=blue,
  linkcolor=blue,
}

\title{From Weakly-terminating Binary Agreement and Reliable Broadcast to Atomic Broadcast}
\author{Andreas Fackler \and Samuel Schlesinger \and Matthew Doty}
\institute{CasperLabs Holdings AG}
\date{}

\newcommand{\GST}{\mathrm{GST}}
\newcommand{\RB}[1]{\mathtt{RB[}#1\mathtt{]}}
\newcommand{\WBA}[1]{\mathtt{WBA[}#1\mathtt{]}}

\newcommand{\N}{\mathbb{N}}

\begin{document}

\maketitle

\begin{abstract}
We present a novel and simple solution to \emph{Atomic Broadcast}
(AB). We reduce AB to two subproblems. One of them is
\emph{Reliable Broadcast} (RB). We also introduce a subproblem we
call \emph{Weakly-terminating Binary Agreement} (WBA). WBA relaxes
\emph{Binary Agreement} (BA) protocols by not always terminating.
WBA admits much simpler solutions than BA. We discuss concrete
solutions to RB and WBA. We prove safety, liveness, and censorship
resilience of our new AB protocol.
\end{abstract}

\section{Introduction}

The core task of an \emph{Atomic Broadcast} (AB) system is to get a network
to agree on a series of transactions. In the case of a database, these
are user submitted database updates. In the case of blockchain, these
are smart contract calls or transfers of digital tokens.

Different contexts lead to different solutions to Atomic Broadcast.
The right solutions depend on how many nodes are allowed to be faulty
and on the speed and reliability of message delivery. They also use
different cryptographic primitives, providing different security models.

Many algorithms reduce Atomic Broadcast to some \emph{broadcast} and
\emph{agreement} subprotocols. In these systems, Atomic Broadcast
repeats the subprotocols over and over. First, one or more nodes
broadcast values they propose for appending to the sequence next.
Then, all nodes reach agreement on those values. After that, the nodes
start over and broadcast again.

Atomic Broadcast protocols often need agreement subprotocols because
ordinary broadcasting can be faulty. A proposing node may be malicious
and send different values to different peers. Nodes can crash. Or the
network could delay messages for very long periods of time.

\emph{Reliable Broadcast}\cite{bracha1987asynchronous} (RB) protocols
prevent malicious proposers and simplify agreement for Atomic
Broadcast. In Reliable Broadcast, \emph{correct} nodes will accept at
most one proposal. But Reliable Broadcast protocols cannot help if a
proposing node is offline. An agreement subprotocol needs to decide which
proposed values to accept.

\emph{Binary Agreement}\cite{bracha_asynchronous_1984} (BA)
complements Reliable Broadcast by deciding which proposals to accept.
In Binary Agreement, all correct nodes will output the same value of 0
or 1 at some point for a given proposal. If they all output 1, then
the Atomic Broadcast protocol adds the proposal to the output
sequence. If they all output 0 then the protocol drops the proposal.

Our contribution is this: we show we can replace Binary Agreement with
a simpler problem. We call the new class of protocol
\emph{Weakly-terminating Binary Agreement} (WBA). In WBA, all correct
nodes will output either 0 or 1, or they will not output at all. With
WBA, if all nodes receive a correct proposal in time then they will
output 1. If they all time out waiting for a correct proposal they
output 0. Otherwise, unlike Binary Agreement, they may not output at
all. This way Binary Agreement protocols are a subclass of
Weakly-terminating Binary Agreement protocols.

Our Atomic Broadcast protocol is easier to implement because it uses
WBA. This is because WBA protocols have fewer requirements than BA
protocols and thus permits simpler solutions which we will present
in this paper.

\section{Network Model}

We consider a distributed network consisting of \emph{nodes} sending
messages to each other. A known subset of $n$ nodes are the
\emph{validators}. A node is \emph{correct} if it executes the
protocols as described, otherwise it is \emph{faulty}. A faulty node
can send any message and ignore the protocols, or even collude with
other faulty nodes. We assume that strictly less than one third of the
validators are faulty, i.e. our fault tolerance $f$ satisfies
$n > 3f$.

A \emph{quorum} is a set of more than $q = \frac{n + f}{2}$
validators. Any two quorums intersect with more than
$f$ shared validators. So they always have a correct
validator in common.

We assume that the network is \emph{partially synchronous}:

All direct messages between correct nodes arrive eventually.
There is a point in time called \emph{Global Stabilization Time}
(\(\GST\)).  After the \(\GST\) messages arrive with maximum delay
\(\delta\). We assume we know \(\delta\) but we do not know \(\GST\).
This assumption is for convenience. In \cite{dwork1988consensus} the
case where we do not know \(\delta\) is considered:
By making all timeouts in the protocol increase over time,
we can accommodate for the unknown delay.

\section{Atomic Broadcast}

Throughout the following sections we describe an Atomic Broadcast
protocol. In an AB protocol some nodes are \emph{proposers}.
The proposers receive multiple inputs of some type \(\mathcal{V}\).
All nodes output multiple values of type \(\mathcal{V}\) so that:

\begin{itemize}
  \item {\bf Agreement:} If any correct node outputs $v$,
    every correct node will eventually output $v$.
  \item {\bf Total Order:} If any correct node outputs $v$ before $w$, all
    correct nodes output $v$ before $w$.
  \item {\bf Censorship Resilience:} Every input $v$ to a correct proposer is
    eventually output by the correct nodes.
\end{itemize}

Together, Agreement and Total Order are summarized as \emph{Safety}.
The additional property that it does not stop outputting values is
\emph{Liveness}, which is implied by Censorship Resilience.

Instead of directly specifying which messages to send, we will solve the AB
problem by using two subprotocols and a timer. Our AB implementation will keep
all inputs it received in a buffer, and proceed by making inputs to and
processing outputs from the subprotocols.

\section{Weakly-terminating Binary Agreement and Reliable Broadcast}

In a \emph{Weakly-terminating Binary Agreement (WBA)} protocol each validator
receives at most one single-bit input (0 or 1) and makes at most one output and:

\begin{itemize}
  \item {\bf Agreement:} If one correct node outputs $b$, all correct nodes
    eventually output $b$.
  \item {\bf Validity:} If the correct nodes output $b$, more than $q - f$
    correct validators had input $b$.
  \item {\bf Weak Termination:} If more than $q$ correct validators have input
    $b$, the correct nodes eventually output $b$.
\end{itemize}

In \emph{Reliable Broadcast (RB)} a designated \emph{proposer} receives one
input and each node makes at most one output and:

\begin{itemize}
  \item {\bf Agreement:} If one correct node outputs $v$, all correct nodes
    eventually output $v$.
  \item {\bf Weak Termination:} If the proposer is correct and has input $v$,
    the correct nodes eventually output $v$.
\end{itemize}

WBA and RB have asynchronous solutions that do not rely on the partial synchrony
assumption. But if there are upper bounds for message delays, these usually give
rise to upper bounds for how much time WBA and RB take:

We say WBA \emph{has delay $\Delta$} if:

\begin{itemize}
  \item If more than $q$ correct validators get the same input before time
    $t \geq \GST$, then all correct nodes output before $t + \Delta$.
  \item If any correct node outputs at time $t \geq \GST$, then all correct
    nodes output before $t + \Delta$.
\end{itemize}

And RB \emph{has delay $\Delta$} if:

\begin{itemize}
  \item If a correct proposer gets input before time $t \geq \GST$, then all
    correct nodes output before $t + \Delta$.
  \item If any correct node outputs at time $t \geq \GST$, then all correct
    nodes output before $t + \Delta$.
\end{itemize}

We present two algorithms for both RB and WBA in section
\ref{sectionRbWbaSolutions}.

\section{Reducing AB to RB and WBA}

\subsection{Idea}

Our Atomic Broadcast solution is leader-based, following
\cite{buchman2018latest,castro1999practical}. The idea is to proceed
in a sequence of rounds. Each round has a designated proposer node
which is the leader. The leader proposes the next value for all nodes
to output. The leader sequence could be pseudorandom or round-robin.
What matters for liveness is that every proposer is the leader of
infinitely many rounds.

We use an RB instance $\RB{r}$ in every round $r$, so it is guaranteed that all
nodes receive the same proposed value --- or none, if the leader is faulty! To
avoid waiting indefinitely for $\RB{r}$ we use a timeout, and a WBA instance
$\WBA{r}$, where every validator inputs $1$ if they receive an acceptable proposal in
time or $0$ if they hit the timeout while waiting.  If $\WBA{r}$ outputs $1$,
that is a decision to \emph{finalize} it, i.e. return it as the next AB output.

Neither RB nor WBA guarantee that they output anything. But the crucial
observation is that if $\RB{r}$ does not output, all the correct validators will
hit the timeout and input $0$ in $\WBA{r}$. In that case, $\WBA{r}$ \emph{does}
output $0$, because of Weak Termination. Thus the protocol can avoid getting
stalled in round $r$ by allowing the next leader in round $r + 1$ to make a
proposal as soon as $\RB{r}$ has output an acceptable proposal \emph{or}
$\WBA{r}$ has output $0$.

The devil is in the details: It can happen that $\RB{r}$ does output an
acceptable proposal and $\WBA{r}$ outputs $0$ anyway (because the timeout hit
first for too many validators). But all nodes need to agree on whether the
round-$r$ proposal should be output before outputting the one from $r + 1$.
That's why the proposal in round $r + 1$ needs to specify whether round $r$
should be skipped.

\subsection{The Algorithm}

For each round $r \in \N$, there is an RB instance $\RB{r}$ with the designated
round-$r$ leader as proposer, and a WBA instance $\WBA{r}$.  Values in $\RB{r}$
are pairs $(v, s)$ that roughly mean: ``I propose outputting $v$ right after the
value from round $s$.'' If $s$ is omitted --- we use the symbol $\bot$ --- that
means it is proposed as the first value. Values in $\WBA{r}$ are $0$ or $1$,
where $0$ means: ``I did not get a proposal in time and vote to allow the next
proposer to skip round $r$.'' And $1$ means: ``I got a proposal and vote for
finalizing it.''

We call round $r$ \emph{committed} if $\WBA{r}$ has output $1$, and
\emph{skippable} if $\WBA{r}$ has output $0$.

For any pair $(v, s)$ with $s \in \N$, we call $s$ the \emph{parent round}; if
$\RB{s}$ has an output, we call that the \emph{parent} of $(v, s)$. The
\emph{ancestors} of $(v, s)$ are its parent and all ancestors of its parent. A
pair $(v, \bot)$ has no parent or ancestors.  If $(v, s)$ with $s \neq \bot$ is
the output of $\RB{r}$, we also call $s$ the \emph{parent} of $r$. The
\emph{ancestors} of a round are its parent and all of its parent's ancestors.

$\bot$ is \emph{fertile} in round $r \in \N$ if all rounds $t < r$ are
skippable. Some $s \in \N$ is \emph{fertile} in round $r$ if $s <
r$, $\RB{s}$ has a fertile output and all rounds $t$ with $s < t < r$ are
skippable. If $\RB{r}$ has a fertile output, then that output is
\emph{accepted} in round $r$.

If $(v, s)$ is accepted in round $r$ and $r$ is committed, then $(v, s)$ and all its
ancestors are \emph{finalized}. We also call $r$ and all of $(v, s)$'s ancestor
rounds finalized. (Note that not all of those ancestor rounds are necessarily
committed.)

The \emph{current round} is the lowest round $r$ which is neither skippable nor
has an accepted value.

We assume that RB and WBA have delay $\Delta$ after $\GST$.

We formulate the protocol by specifying what actions to take whenever certain
conditions become true in a node $N$. These need to be checked whenever the
timer fires or one of the subprotocols outputs. When we write ``input'' in some
RB/WBA, we mean ``input unless we have already made an input earlier''.

\begin{itemize}
  \item When $N$ is leader in $r$, and has an input $v$ that has not been
    finalized yet, and there is a fertile $s$ in $r$,
    it inputs $(v, s)$ in $\RB{r}$.
  \item When a new round $r$ becomes current, $N$ cancels and restarts
    the timer, with delay $2 \Delta$.
  \item When the timer fires and $r$ is current, $N$ inputs $0$ in $\WBA{r}$.
  \item When there is an accepted value in a round $r$, $N$ inputs $1$ in
    $\WBA{r}$.
  \item When new values become finalized, $N$ outputs them, lower rounds first.
\end{itemize}

\subsection{Pseudocode}

For simplicity we assume that all subprotocols ignore unexpected inputs.
Equivalently, S.{\tt input}(x) means we input $x$ into subprotocol $S$, \emph{if}
it expects an input, i.e. if we have not made an input before and, in the case of
RB, if we are the designated proposer. Outside the subprotocols, only three
variables are needed. Upon startup, we initialize the timer and call the main handler:

\begin{lstlisting}[language = pseudopython]
def initialize():
  current: $\mathbb{N}$ = 0
  undecided_round: $\mathbb{N}$ = 0
  inputs: List[$\mathcal{V}$] = []

  start_timer(2 * $\Delta$)
  on_subprotocol_output()
\end{lstlisting}

We first implement some definitions and a helper function for finalizing values
in the right order:

\begin{lstlisting}[language = pseudopython]
def fertile($r$: $\mathbb{N}$, $s$: Option[$\mathbb{N}$]) $\to$ $\mathbb{B}$:
  if $s$ == $\bot$:
    return $\forall t < r.\;$ WBA[$t$].output == 0
  else:
    return $s < r \;\wedge\;(\forall u.\; s < u < r \longrightarrow $ WBA[$u$].output == 0$) \;\wedge\; \exists t.$ accepted($s$, $t$)

def accepted($r$: $\mathbb{N}$, $s$: Option[$\mathbb{N}$]) $\to$ $\mathbb{B}$:
  return fertile($r$, $s$) $\wedge$ $\exists v.$ RB[$r$].output == $(v, s)$

def finalize($r$: $\mathbb{N}$):
  ($v$, $t$) = RB[$r$].output
  if $t \neq \bot$ $\wedge$ $t$ $\geq$ undecided_round:
    finalize(t)
  inputs.remove($v$)
  output($v$)
\end{lstlisting}

Finally the main part of the protocol: a handler for inputs, i.e. user-submitted
transactions, one for the timer, and one for whenever any WBA or RB subprotocol
instance outputs a value.

\begin{lstlisting}[language = pseudopython]
def on_input($v$: $\mathcal{V}$):
  inputs = [$v$] + inputs
\end{lstlisting}

\begin{lstlisting}[language = pseudopython]
def on_timeout():
   WBA[current].input(0)
\end{lstlisting}

\begin{lstlisting}[language = pseudopython]
def on_subprotocol_output():
  if RB[current].output $\neq \bot$ $\vee$ WBA[current].output == 0:
    current = current + 1
    restart_timer(2 * $\Delta$)

  if len(inputs) > 0 $\wedge$ $\exists s \in$ Option[$\mathbb{N}$]. fertile(current, $s$):
    RB[current].input((inputs[0], $s$))

  if $\exists s \in \N, t \in$ Option[$\mathbb{N}$]. accepted($s$, $t$):
    WBA[$s$].input(1)

  if $\exists s \in \N, t \in$ Option[$\mathbb{N}$]. ($s$ $\geq$ undecided_round $\wedge$ accepted($s$, $t$) $\wedge$ WBA[$s$].output == 1):
    finalize($s$)
    undecided_round = $s + 1$
\end{lstlisting}

\subsection{Proofs}

We described the algorithm from the point of view of one node, or of someone who
implements it. The proofs concern the behavior of a whole network, so we have to
distinguish between the protocol states in different nodes at different times.
We write $r$ is \emph{$(N, t)$-committed} to say node $N$ at time $t$ sees round
$r$ as committed. We write $r$ is \emph{$t$-committed} to say it is $(N,
t)$-committed for every correct node $N$. Analogous definitions apply to the
notions \emph{skippable}, \emph{accepted} and \emph{finalized}.

\begin{lemma}\label{lemmaCommitted}
Let $N$ be a correct node, and let round $r$ be $(N, t)$-committed.
\begin{enumerate}
\item\label{lemmaFutureCommitted} Then $r$ is $(N, t')$-committed for all $t' > t$.
\item\label{lemmaAgreeCommitted} There is a $t'$ such that $r$ is $(N', t')$-committed for all correct nodes $N'$.
\item\label{lemmaDelayCommitted} If RB and WBA have delay $\Delta$ and $t \geq \GST$ then $r$ is
      $(t + \Delta)$-committed.
\end{enumerate}
The same holds true for the properties \emph{skippable}, \emph{accepted} and \emph{finalized}.
\end{lemma}

\begin{proof}
We only prove each statement for the committed case. A round $r$ is $(N, t)$-committed
when $\WBA{r}$ has output $1$. By definition, RB and WBA only output one item, showing
\eqref{lemmaFutureCommitted}. By the Agreement property of WBA, all correct nodes
will eventually see the same output, which shows \eqref{lemmaAgreeCommitted}.
Assuming we are after GST, and RB and WBA only take time $\Delta$, we know that
every correct node $N'$ will be able to infer that same commitment within time
$\Delta$, as that happens through RB and WBA, proving \eqref{lemmaDelayCommitted}.
\end{proof}

The only aspect in which different correct nodes' views can differ is the order
in which these properties are seen. However, finality was defined such that it
does have a monotonicity property anyway, i.e. if a round $r_2$ becomes
finalized later than $r_1$, then $r_2 > r_1$:

\begin{lemma}
Let $N$ be a correct node.

If both $r_1$ and $r_2$ are $(N, t)$-finalized they are either equal or
ancestors of each other.

If $r_1$ is $(N, t_1)$-finalized and $r_2$ is $(N, t_2)$-finalized but not $(N,
t_1)$-finalized then $t_2 > t_1$, $r_2 > r_1$ and $r_1$ is an ancestor of $r_2$.
\end{lemma}

\begin{proof}
For the first claim we omit the prefix $(N, t)$ since it's only about one node
and one point in time.

That $r_1$ is finalized means that it is equal to or an ancestor of some round
$r'_1$ that has accepted a value and is committed. Similarly $r_2$ is equal to
or ancestor of a committed $r'_2$ with accepted value. If $r'_1 = r'_2$ then
both $r_1$ and $r_2$ are ancestors of that round and therefore equal or
ancestors of each other. So assume now that $r'_1 \neq r'_2$, w.l.o.g. $r'_2 >
r'_1$. Let $s$ be minimal among $r'_2$ and its ancestors such that $s > r'_1$.
By the recursive definition of \emph{accepted}, since $r'_2$ has an accepted value,
so does $s$. Let $(v, s')$ be accepted in $s$. If $s' > r'_1$ that would
contradict the minimality of $s$. If $s' < r'_1$ that would mean $r'_1$ must be
skippable, which is impossible since $r'_1$ is committed. Therefore $r'_1 = s'$.
Hence $r'_1$ is also an ancestor of $r'_2$. Thus both $r_1$ and $r_2$ are
ancestors of or equal to $r'_2$, so they are ancestors of or equal to each
other.

For the second claim note that $t_2 > t_1$ because the property of being
finalized can only become true with more RB/WBA outputs arriving, not false
again.

In particular both $r_1$ and $r_2$ are $(N, t_2)$-finalized. By the first part
of the lemma that means they are ancestors of each other. If $r_1$ were greater
than $r_2$, $r_1$ by definition could not be $(N, t_1)$-finalized without $r_2$
also being $(N, t_1)$-finalized. Hence $r_2 > r_1$.
\end{proof}

In other words the rounds are observed as finalized in increasing order; it
cannot happen that a node first sees $r$ as finalized, and then later $s < r$.
It follows from the Agreement property of RB and WBA that the set of rounds that
are eventually finalized is the same in all nodes. Hence in every correct node
the $k$-th output is exactly the $k$-th element in the set of all rounds that
eventually get finalized. That proves:

\begin{theorem}[Safety; Agreement and Total Order]
If any correct node's $k$-th output is $v$, then every correct node will
eventually output $k$ values and the $k$-th one is $v$.\qed
\end{theorem}

So far we did not use the partial synchrony assumption, which indeed is not
needed for safety. From now on assume that RB and WBA have delay $\Delta$.

\begin{lemma}\label{lemmaCurrentRound}
If all correct nodes start before $\GST$, then for every $r$, all rounds $s < r$
are $(\GST + 3 r \Delta)$-skippable or have a $(\GST + 3 r \Delta)$-accepted
value. In particular, at time $\GST + 3 r \Delta$, every correct node's current
round is $\geq r$.
\end{lemma}

\begin{proof}
We show this by induction on $r$.

The base case is trivial since there are no rounds before the first one.

So let $r > 0$ and $t' = \GST + 3 r \Delta$, and assume the induction
hypothesis: All rounds $s < r - 1$ are $(t' - 3 \Delta)$-skippable or have a
$(t' - 3 \Delta)$-accepted value.

That means all correct validators start their timer for round $r - 1$ before $t'
- 3 \Delta$.

If round $r - 1$ has a $(N', t' - \Delta)$-accepted value for some correct $N'$,
it has a $t'$-accepted value.

Otherwise all correct validators will hit the timeout before $t' - \Delta$ and
input $0$ in $\WBA{r - 1}$, so WBA outputs $0$ before $t'$ and $r - 1$ is
$t'$-skippable.

By Lemma \ref{lemmaCommitted}, all rounds $s < r - 1$ are also $t'$-skippable or
have a $t'$-accepted value, proving the statement for $r$.
\end{proof}

\begin{lemma}\label{lemmaFinalization}
Let $R$ be the highest round that is current in any correct node at GST. Every
round $r > R$ with a correct leader that has an unfinalized input eventually has
an accepted value and becomes committed.
\end{lemma}

\begin{proof}
By Lemma \ref{lemmaCurrentRound}, all correct nodes eventually reach a round
$\geq r$. Let $N$ be the first correct node to do so, i.e. the node with the
minimal $t$ such that all rounds $s < r$ are $(N, t)$-skippable or have an $(N,
t)$-accepted value. So no correct node starts its timer for round $r$ before
time $t$ and no correct node inputs $0$ in $\WBA{r}$ before $t + 2 \Delta$.

By our assumption $t > \GST$, so by Lemma \ref{lemmaCommitted} all correct nodes
reach round $\geq r$ before $t + \Delta$.

In particular if the leader $L$ in round $r$ is correct and has an input $v$
that is not finalized yet, it will propose it before $t + \Delta$, and $\RB{r}$
will output it in all correct nodes before $t + 2 \Delta$. Since the parent of
the proposal is $(L, t + \Delta)$-accepted, it will be $(t + 2
\Delta)$-accepted. Thus the proposal itself is also $(t + 2 \Delta)$-accepted,
and each correct validator inputs $1$ in $\WBA{r}$, so $r$ becomes committed.
\end{proof}

The requirement for the leader sequence was that every proposer is the leader
infinitely many times, so Lemma \ref{lemmaFinalization} implies:

\begin{theorem}[Liveness; Censorship Resilience]\label{thmLiveness} Every value
input in a correct proposer is eventually output by the correct nodes. \qed
\end{theorem}

\section{Reliable Broadcast and Weakly-terminating Binary Agreement Solutions}\label{sectionRbWbaSolutions}

\subsection{Bracha's Algorithm}\label{sectionBracha}

Bracha presents a simple implementation of Reliable Broadcast
in \cite{bracha1987asynchronous}. His algorithm may be summarized as
follows:

\begin{enumerate}
  \item On input $v$, the proposer sends $(\mathtt{initial}, v)$ to all nodes.
  \item Each correct validator waits for either one $(\mathtt{initial}, v)$ from
  the proposer, or a quorum of $(\mathtt{echo}, v)$ or $(\mathtt{ready}, v)$
  from $> f$ validators, then sends $(\mathtt{echo}, v)$ to everyone.
  \item Each correct validator waits for either a quorum of $(\mathtt{echo}, v)$
  or $(\mathtt{ready}, v)$ from $> f$ validators, then sends $(\mathtt{ready},
  v)$ to everyone.
  \item Each correct node waits for $(\mathtt{ready}, v)$ from $> 2f$
  validators, then outputs $v$.
\end{enumerate}

A slight change turns this into a solution for the WBA problem as well:

\begin{enumerate}
  \item Each correct validator waits for either an input $b$ from
  the proposer, or a quorum of $(\mathtt{vote}, b)$ or $(\mathtt{ready}, b)$
  from $> f$ validators, then sends $(\mathtt{vote}, b)$ to everyone.
  \item Each correct validator waits for either a quorum of $(\mathtt{vote}, b)$
  or $(\mathtt{ready}, b)$ from $> f$ validators, then sends $(\mathtt{ready},
  b)$ to everyone.
  \item Each correct node waits for $(\mathtt{ready}, b)$ from $> 2f$
  validators, then outputs $b$.
\end{enumerate}

The $\mathtt{vote}$s serve the same purpose as the $\mathtt{echo}$s, but instead
of the value received from the proposer, they contain the sender's own input
value.

These algorithms require a network where all correct validators are directly
connected to all correct nodes, but work without any synchrony assumptions.

If the maximum message delay is $\delta$, they have delay $3 \delta$ resp. $2
\delta$.

\subsection{Gossiping Quorums of Signatures}\label{sectionGossipQuorum}

A different solution that works even in larger networks where a validator cannot
expect to be directly connected to all other correct nodes is using a gossip
mechanism to disseminate signatures and wait for a quorum. Gossip protocols are
a broad subject themselves, so assume we have one that satisfies the following
properties:
\begin{itemize}
  \item If a correct node gossips a message $m$, eventually all correct nodes receive $m$.
  \item If any correct node receives a message $m$, eventually all correct nodes receive $m$.
\end{itemize}
Note that even if all correct nodes were connected to each other, this is a
stronger guarantee than we get by just sending $m$ to everyone: If a faulty node
sends $m$ to some and $m'$ to other correct nodes, a gossip protocol guarantees
that all correct nodes will receive both $m$ and $m'$.

We also need cryptographic signatures, and a setup where each validator has a
public key known to all nodes. In the following we write $(\ldots, \sigma)$ for
a message where $\sigma$ is a signature of the other fields.

Given these tools, a solution to RB is simply:
\begin{enumerate}
  \item On input $v$, the proposer signs and gossips $(\mathtt{initial}, v, \sigma)$.
  \item When a correct validator receives $(\mathtt{initial}, v, \sigma)$ signed
  by the proposer, and has not signed an $\mathtt{echo}$ yet, it signs and
  gossips $(\mathtt{echo}, v, \sigma')$.
  \item When a correct node receives $(\mathtt{echo}, v, \sigma_i)$ with a set
  of signatures $\sigma_i$ from a quorum of validators, it outputs $v$.
\end{enumerate}

And WBA:
\begin{enumerate}
  \item On input $b$, a correct validator signs and gossips $(\mathtt{vote}, b,
  \sigma)$.
  \item When a correct node receives $(\mathtt{vote}, b, \sigma_i)$ with a set
  of signatures $\sigma_i$ from a quorum of validators, it outputs $b$.
\end{enumerate}

To see that these algorithms indeed solve the RB/WBA problems, observe that each
correct validator only ever signs one $\mathtt{echo}$ or $\mathtt{vote}$ message,
and since any two quorums overlap in at least one correct validator, there can be
a quorum for at most one value. By our assumptions about the gossip
protocol, if any correct node sees such a quorum, all of them will eventually
see it. That implies RB and WBA Agreement. WBA Validity follows because a quorum
must contain signatures from at least $q - f$ correct validators. RB and WBA Weak
Termination follow because the gossip protocol guarantees to deliver a correct
proposer's unique $\mathtt{initial}$ message to everyone, and a correct
validator's $\mathtt{echo}$ or $\mathtt{vote}$.

The delay of WBA is just the time the gossip algorithm takes to deliver a message
to all correct nodes, and the delay of RB is twice that.

By our assumption about gossip, if any correct node receives the
$\mathtt{initial}$ message, all of them will. If the value $v$ is large, the RB
algorithm can thus easily avoid making lots of redundant copies by replacing the
$v$ in the $\mathtt{echo}$ messages with a cryptographic hash of $v$. Point 3
then has to be modified slightly, since a node can only output $v$ once it has
received a quorum of $\mathtt{echo}$s \emph{and} the $\mathtt{initial}$ message.

\section{Practical Considerations}

\subsection{Censorship Resilience with Unknown \texorpdfstring{$\delta$}{Delta}}

The more realistic version of partial synchrony is with an unknown maximum
message delay $\delta$, which will result in unknown RB and WBA delays $\Delta$.
In theory all protocols can be adapted to that version by slowing down their
clocks over time: the later a timer is started the more its delay is increased
compared to the known-$\delta$ variant.

In practice, however, ever-increasing timeouts would mean that the user-visible
delay caused by a single crashed or faulty proposer becomes longer and longer.
So implementations usually reset the timeout back to a lower delay whenever a
new value is finalized. If the timeout is reset to a very low value after each
finalization, this breaks the proof of Theorem \ref{thmLiveness}: It could be
that not every (or even no) correct block proposer gets their proposals
finalized, because the timer is always reset before it is their turn.

One strategy is to keep increasing the timeout until at least $k$ out of $p$
proposers got their most recent round committed. That achieves a lower level of
censorship resilience: At least $k$ proposers will eventually get all their
inputs finalized. But if more than $p - k$ proposers are faulty, it will lead to
the timeouts and user-visible delays increasing indefinitely.

Another approach is to modify the protocol: A round $r$ is \emph{orphaned} if it
has accepted a value but there is a finalized round $s > r$ of which $r$ is not
an ancestor. So an orphaned round is one that we know will never get finalized.
Now instead of tuples we use triples $(v, r, s)$ as proposal values: The third
entry specifies an orphaned \emph{uncle} round, and if the proposal gets
finalized, we output both the value accepted in $s$ and the value $v$. $(v, r,
s)$ is only acceptable if $s$ is $\bot$ or has an accepted value. But in
addition, we only input $1$ in WBA if we either do not have an orphaned round or
$s$ points to the lowest orphaned round we know of. This effectively forces
proposers in later rounds to acknowledge each orphaned value as an uncle and
indirectly finalize it, too. Thus the algorithm remains fully censorship
resilient even if the timeout is reset whenever a round is finalized. Note that
this complicates the liveness proof, but since eventually all correct nodes will
agree on which is the lowest orphaned round, it still works.

\subsection{Validation}

The validity of a proposed value often depends on its ancestors in practice:
E.g. a client request to a database should only be executed once, so the
proposal $(v, r)$ is only valid if none of its ancestors already contains the
value $v$. Or a smart contract call is only valid if the caller has enough funds
to pay the transaction fees. Or a block in a blockchain must contain its
parent's hash.

In order to account for these various details, one might extend the definition of
accepted in order to only accept valid values. Be careful to avoid proposing
values which would be considered invalid in your implementation, of course.

\subsection{Spam}

In theory the protocol runs an infinite number of RB and RV instances. But
correct nodes make inputs to $\RB{r}$ or $\WBA{r}$ once a round $\geq r$ is
current. Implementation should provide a way for correct nodes to reject
incoming messages belonging to implausibly high rounds, e.g. by making the
sender queue them until they are ready.

\subsection{Minimum Delays}

The protocol can be extended to respect a configured start time and a minimum
delay between a value and its parent by adding two additional conditions for
incrementing our round number and starting the timer: Round $r$ becomes current
when all rounds $s < r$ are skippable or have an accepted value \emph{and}:
\begin{itemize}
  \item The current time is at least the configured start time for the network.
  \item At least the minimum delay has passed since the highest timestamp of
    any accepted value.
\end{itemize}

\section{Comparison with other Protocols}\label{sectionComparison}

The protocol discussed in this paper has similarities to many other known
protocols. HoneyBadgerBFT \cite{miller2016honey}, DBFT \cite{crain2018dbft} and
Aleph \cite{gkagol2019aleph} are examples of leaderless protocols that also use
RB as a subprotocol.

HoneyBadgerBFT and DBFT also use full BA, which guarantees termination in all
cases.  As these protocols are leaderless, instead of one RB and BA instance per
round, there is one instance per round \emph{per proposer}.  Assuming enough
proposers are correct, it is guaranteed that a certain number of RB instances
will terminate. This replaces the timer: Validators input $1$ in a BA if the
corresponding proposer's RB was among the first to deliver a value, otherwise
they input $0$. HoneyBadgerBFT then combines all proposed values where BA output
$1$, whereas DBFT just uses one of them. DBFT uses a BA solution that requires
partial synchrony, whereas HoneyBadgerBFT is asynchronous thanks to a BA
protocol based on threshold cryptography.

Aleph uses an asynchronous threshold cryptography-based BA in addition to RB,
but reduces communication complexity by arranging the proposed values in a
directed acyclic graph (DAG) where every proposal points to all proposals the
sender has seen before, and works on an implicit message in several BA instances
simultaneously.

While they share with our protocol the property of making subprotocols explicit,
these three algorithms are vastly more complex than ours, which belongs to the
family of leader-based partially synchronous protocols following the Lock-Commit
paradigm \cite{abraham2020lock}. It can be compared to a number of other members
of that family, e.g. HotStuff \cite{yin2018hotstuff}, SBFT \cite{gueta2019sbft},
Doomslug \cite{skidanov2019doomslug}, PBFT \cite{castro1999practical},
Tendermint \cite{buchman2018latest}. In this context, we will use the terms
\emph{block} and \emph{value} synonymously.

SBFT trades a slightly lower fault tolerance for a fast-path that finalizes
blocks in just one round of messages if \emph{almost all} validators are
correct. Doomslug also finalizes blocks in a single round, but only with fault
tolerance 1, and then later finalizes all blocks with a higher fault tolerance.

HotStuff and Libra support \emph{pipelining}, similar to Aleph: To reduce
overall communication complexity per block and to minimize the time between
subsequent blocks, consensus for one of them takes several subsequent rounds to
be reached. The downside is that the time to consensus is tied to how long it
takes to broadcast the next proposals, as well as to any minimum time between
blocks that one may want to enforce in practice.

PBFT and Tendermint do not do pipelining but try to finalize each block before
the next one, which requires multiple rounds of messages.

By allowing the next proposal as soon as RB has output, but having a separate
WBA instance per round, we do something in between: Very few RB messages have to
be exchanged between subsequent blocks, but WBA can finalize the value without
waiting for later rounds.

Apart from that, if the gossip-based RB and WBA implementations from section
\ref{sectionGossipQuorum} are used, Tendermint is closest to our protocol. So we
will compare them in detail:

First of all, the Tendermint paper describes a setting where a new instance of
the protocol is run for every block, and the protocol outputs only once. But
it can easily be adapted to make multiple outputs: Instead of re-proposing the
valid value (see below), a proposer could create a new child of that value. For
the purpose of our comparison we will work with that version.

In Tendermint, every round is subdivided into three phases:

\begin{itemize}
  \item First the leader gossips a $\mathtt{proposal}$, which also includes a
  value and the intended parent round --- the proposer's highest \emph{valid
  round} (see below).
  \item Then all validators send a $\mathtt{prevote}$ with the hash of the
  proposal they received, or with $\mathtt{nil}$ if they received none.
  \item Finally the validators send a $\mathtt{precommit}$ with the hash of the
  proposal they saw a quorum of $\mathtt{prevote}$s for, or $\mathtt{nil}$ if
  none.
\end{itemize}

Each phase has its own timeout: the first one is started when the round begins,
but the prevote resp. precommit timeouts start when a quorum of
$\mathtt{prevote}$ resp. $\mathtt{precommit}$ messages are seen, even if their
content does not match.

The first two phases clearly correspond to RB, $\mathtt{proposal}$ to
$\mathtt{initial}$  and $\mathtt{prevote}$ to $\mathtt{echo}$. However there is
no analog to a $\mathtt{nil}$ $\mathtt{prevote}$ in our protocol. Instead, RB
simply does not guarantee to terminate, and there is one single timeout for it.

The third phase corresponds to WBA, $\mathtt{precommit}$ to $\mathtt{vote}$, the
proposal hash to $1$ and $\mathtt{nil}$ to $0$. The difference here is that
Tendermint has a timer for this phase and the next proposer has to wait for
another quorum or timeout even if there is a quorum of $\mathtt{prevote}$s.

The \emph{valid round} in Tendermint is the highest one in which we have seen a
quorum of $\mathtt{prevote}$s for the same proposal. It corresponds to the
highest round with an accepted value. The \emph{locked round} in a Tendermint
node is the highest round for which \emph{that node} has sent a
$\mathtt{precommit}$ for a proposal. This is subjective and replaces the
skipping rule: A node will refuse to $\mathtt{prevote}$ for a new proposal if
our locked round is between it and its parent round.

\section{Conclusion}

We believe that isolating the RB and WBA subprotocols makes this partially
synchronous Atomic Broadcast protocol particularly easy to understand, and gives
intuitive meaning to its messages, without compromising on desirable properties
regarding performance and fault tolerance.

\bibliographystyle{plain}
\bibliography{consensus}

\appendix

\end{document}